\documentclass[12pt]{amsart}
\usepackage[letterpaper, margin=1in]{geometry}
\usepackage{url}
\usepackage{hyperref}
\usepackage{amsmath}
\usepackage{amssymb}
\usepackage{amscd}
\usepackage{algorithm}
\usepackage{xcolor}
\usepackage{algorithmic}

\usepackage{tikz}
\newtheorem{theorem}{Theorem}[section]
\newtheorem{lemma}[theorem]{Lemma}
\newtheorem{corollary}[theorem]{Corollary}
\newtheorem{proposition}[theorem]{Proposition}
\theoremstyle{remark}

\newtheorem{remark}[theorem]{Remark}
\numberwithin{equation}{section}
\newtheorem{definition}[theorem]{Definition}

\begin{document}
\title{Worst--Case to Average--Case Reductions for SIS over integers}
\author{Konstantinos A. Draziotis}
\email{drazioti@csd.auth.gr}
\author{Myrto Eleftheria  Gkogkou}
\email{myrtoele@csd.auth.gr}
\address{School of Computer Science, Aristotle University of Thessaloniki, Thessaloniki, Greece}
\thanks{Both authors were co-funded by SECUR-EU. The SECUR-EU project funded under Grant Agreement 101128029 is supported by the European Cybersecurity Competence Centre.}%

\maketitle
\begin{abstract}

 In the present paper we study a non-modular variant of the Short Integer Solution problem over the integers. Given a random matrix $A \in \mathbb{Z}^{n\times m}$ with entries  $a_{ij}$ such that $0\le a_{ij}< Q,$ for some $Q>0,$ the goal is to find a nonzero vector ${\bf x}\in\mathbb{Z}^m$ such that $A{\bf x}={\bf 0}$ and $\|{\bf x}\|_\infty \le \beta,$ for a given bound $\beta.$ We show that an algorithm that solves random instances of this problem with non-negligible probability yields a polynomial-time algorithm for approximating $\mathrm{SIVP}$ within a factor $\widetilde{O}(n^{3/2})$ (with $\ell_2$ norm) in the worst case for any $n-$dimensional integer lattice. 

\end{abstract}

\section{Introduction}
Since Ajtai’s seminal work~\cite{ajtai96}, which established the first connection between worst-case and average-case lattice problems, substantial progress has been made in constructing cryptographic systems grounded in worst-case lattice assumptions. The significance of this connection lies in its potential to enable the use of lattices as a foundation for creating robust cryptosystems.  For example, NIST has initiated a process to solicit, evaluate, and standardize one or more quantum-resistant public-key cryptographic algorithms. 

Lattice-based cryptographic schemes present several key advantages. Firstly, they are highly efficient and straightforward, typically relying on basic linear operations with small integers. Secondly, unlike cryptographic methods based on factoring or discrete logarithms, they have shown resilience against quantum attacks. Lastly, they offer strong security guarantees by ensuring that random instances are as hard as some worst-case problems in lattices.

\subsection{Our Contribution}
In the present work we introduce an average-case problem that, if solvable, implies a solution to a hard lattice problem. 
More precisely, we consider a variant of plain SIS, but we work on the ring ${\mathbb{Z}}$ instead of ${\mathbb{Z}}_q.$ The average problem is $A{\bf X}={\bf 0}$ for $A$ an $n\times m$ random matrix and we search for ${\bf X}\in {\mathbb{Z}}^m$ such that $\| {\bf X}\|_{\infty}\leq\beta$ (for some $\beta>0$ usually small). The norm $\ell_\infty$ is defined by $\|{\bf X}\|_\infty=\max_{1\le j\le m}\{|X_j|\}.$  If we can solve the previous for uniform matrices $A,$ whose entries are chosen from some finite set ${\mathcal{S}}\subset {\mathbb{Z}}$, we can solve $SIVP_{\tilde{O}(n^{1.5})}$ (with $\ell_2$ norm) for any integer lattice in the worst case. We prove the following Theorem.
\begin{theorem}\label{theorem}
    For an integer \( m = O(n^2) \), if there exists a polynomial-time probabilistic algorithm that solves $\ell_{\infty}-{\rm{SIS}}_{\mathbb{Z}}$ for uniformly random \( A \in {\mathcal{C}}_Q^{n\times m} \) ($Q$ is a positive integer, ${\mathcal{C}}_Q=\{0,1,...,Q-1\})$, then the SIVP problem can be approximated in polynomial time within a factor of 
    \(\tilde{O}(n^{1.5}) \) in any \( n \)-dimensional integer lattice.
\end{theorem}

The proof cannot be obtained by a black-box reuse of the standard ${\rm{SIS}}_q$ reduction (Subsection \ref{no-black-box-reuse}). The reduction follows the well-known paradigm used for the standard SIS problem but introduces several important modifications. For instance, we employ Siegel’s Lemma to bound integer solutions of linear systems.  Unlike the modular setting, where the condition ``$\bmod q$'' naturally constrains the entries of the matrix $A$ (the input of the ${\rm{SIS}}_q$ oracle) to be within $[0, q)$, in our case we must explicitly define an integer $Q$ that bounds $A$'s entries. Moreover, our approximation factor is $\tilde{O}(n^{1.5})$ instead of $\tilde{O}(n)$ as in  ${\rm{SIS}}_q$.

Overall, the paper develops the hardness theory of this integer SIS variant \footnote{For ${\rm LWE}_q$ the situation is different. As shown in~\cite{espitau}, the corresponding non-modular variant can be solved efficiently, provided that the noise parameters satisfy certain plausible conditions.}.

\subsection{Related Work}

In Ajtai's landmark paper \cite{ajtai96}, a family of one-way hash functions was introduced, with security rooted in the Shortest Independent Vector Problem (SIVP). Shortly thereafter, in \cite{oded}, the authors extended Ajtai's construction to define a collision-resistant hash function. Essentially, this function adapts Ajtai's original construction, restricting it to a more specific definition set to achieve collision resistance. 

Previous constructions suffer from inefficiency due to the need for handling and storing a large matrix. Addressing this, Micciancio introduced an alternative function in \cite{Mic-focs2002} of the form 
\[
f_{\mathbf{a}}(\mathbf{x}) = \mathbf{a} \cdot \mathbf{x} \pmod{q},
\] 
with  convolution product operation (as opposed to componentwise multiplication in Ajtai's construction).
This problem is computationally hard to invert when restricting the solution set to \( \{0,1,\ldots,\lfloor p^{\delta} \rfloor\}^n \). He demonstrated that solving random instances of this problem with non-negligible probability is as hard as approximating the Shortest Vector Problem (SVP) for \(\tilde{O}(n)\)-cyclic lattices in the worst case, with a complexity nearly linear in the lattice dimension. Later,  Micciancio and Lyubashevsky~\cite{Mic07} advanced the previous approach by developing efficient, collision-resistant hash functions with security grounded in standard lattice problems, specifically tailored for ideal lattices.

Furthermore, other cryptographic primitives, e.g., public key cryptosystems, ID schemes, digital signatures etc., have also emerged based on lattice problems. One year after Ajtai's work, Ajtai and Dwork~\cite{ajtai-dwork} presented the first public key cryptosystem with provable security, relying on a worst-case lattice problem. Shortly thereafter, the GGH cryptosystem \cite{GGH} was developed. Although more efficient, GGH lacked a worst to average-case reduction.  The security of GGH relies on the difficulty of solving CVP using a bad basis of the lattice. This bad basis serves as the public key, requiring $O(n^2)$ space for the size of the public key. For small dimensions (e.g., less than 100), GGH is not secure since the LLL algorithm can break it. For moderate dimensions, Nguyen~\cite{broken-GGH} provided an efficient attack. Therefore, to ensure security we need $n \approx 500$, which renders GGH impractical. In parallel with GGH, the NTRU cryptosystem was introduced~\cite{NTRU}, addressing the issue of large key sizes in the high-dimensional space required by GGH, by using $\tilde{O}(n)$ bits instead of ${O}(n^2)$ for the size of the public key. NTRU has been standardized by IEEE, X9.98, and PQCRYPTO, and was a finalist in the NIST post-quantum cryptography standardization effort.

Also we have other cryptographic schemes based on hard problems on lattices. In 2008 \cite{Vadim}, an identification scheme based on the SIS problem over ${\mathbb{Z}}_q$  was introduced. 

Finally, an alternative approach was used for the identification scheme based on the compact knapsack problem\footnote{i.e. find small integer solutions to the linear equation $a_1x_1+\cdots+a_nx_n=b$.} in \cite{papadopoulou} and later generalized to include also digital signatures in \cite{rizos}. However, no worst-case to average-case reduction was provided for these schemes. 

\subsection{Roadmap}
Section~\ref{sec:lattices} introduces the notation and recalls the lattice background needed throughout the paper. In Section~\ref{no-black-box-reuse} we explain why the standard worst-to-average reduction for $\mathrm{SIS}_q$ cannot be reused as a black box in the non-modular setting. Section \ref{sec:basic-lemmata} collects the lemmas needed for our worst-to-average-case reduction, including standard results on the smoothing parameter, largely summarizing established findings from the literature. Additionally, we cover fundamental aspects of Siegel’s lemma for linear systems. In Section \ref{Sec:reduction} we provide the worst-to-average-case reduction. Finally, in Section \ref{sec:conclusion} we provide a conclusion. The paper also includes two appendices.

\section{Preliminaries}\label{sec:lattices}
In this section, we recall some well-known facts about lattices and provide some definitions.
\subsection{Notation}
In this paper, matrices are denoted by capital letters $A, B, C$; the matrix $B$ typically denotes the basis matrix of a lattice, while $A$ typically refers to the coefficient matrix of a homogeneous system $A\mathbf{x} = \mathbf{0}$. Vectors are written in bold. We denote by $\|\cdot\|$ the Euclidean $\ell_2$ norm. 

We also employ standard asymptotic notation: Given two functions $f, g \colon \mathbb{N} \to \mathbb{R}_{>0}$, we write $f(n) = O(g(n))$ if there exist constants $c > 0$ and $n_0 \in \mathbb{N}$ such that $f(n) \leq c\, g(n)$ for all $n \geq n_0.$ We write $f(n) = \Theta(g(n))$ if $f(n) = O(g(n))$ and $g(n) = O(f(n)).$ We write $f=\Omega(g),$ if $f$ is bounded below from $g$ up to a constant factor i.e. $f(n)\ge k g(n)$ for some $k>0, n_0\in {\mathbb{N}}$ and for all $n\geq n_0.$ Similarly, we write $f(n) = \widetilde{O}(g(n)),$ if there exists a constant $k > 0$ such that, 
$f(n) = O(g(n) \log^k g(n))$, that is, $\widetilde{O}(\cdot)$ suppresses polylogarithmic factors. Similarly for $\tilde{\Theta}$ and $\tilde{\Omega}.$

A function \( f: \mathbb{N} \to \mathbb{R} \) is negligible if for all constants \( c > 0 \), there exists an integer \( N_c \) such that for all \( n \geq N_c \),
\[
|f(n)| < \frac{1}{n^c}.
\]
Also, we need the definition of the statistical distance.
\begin{definition}\label{def:statistical_distance}
    If \( X \) and \( Y \) are discrete random variables over a countable set \( A \), then the statistical distance between \( X \) and \( Y \), denoted \( \Delta(X, Y) \), is defined as
\[
\Delta(X, Y) = \frac{1}{2} \sum_{a \in A} \big| \Pr[X = a] - \Pr[Y = a] \big|.
\]
If \( X \) and \( Y \) are continuous random variables over ${\mathbb{R}}^n$, with density functions $D_X,D_Y$ (resp.), then the statistical distance between \( X \) and \( Y \), again denoted \( \Delta(X, Y) \), is defined as
\[
\Delta(X, Y) = \frac{1}{2} \int_{\mathbb{R}^n} |D_X(t)-D_Y(t)| \ dt.
\]
\end{definition}

\subsection{Lattices}\label{sec:prelim}
Let ${{\bf{b}}_1,{\bf{b}}_2,\ldots,{\bf{b}}_n}$ be linearly independent vectors of ${\mathbb{R}}^{m}$.
	The set 
	\[\mathcal{L} = \bigg{\{} \sum_{j=1}^{n}\alpha_j{\bf{b}}_j :
	\alpha_j\in\mathbb{Z}, 1\leq j\leq n\bigg{\}}\]
	is called {\em lattice} and 
	the finite vector set $\mathcal{B} = \{{\bf{b}}_1,\ldots,{\bf{b}}_n\}$ is called basis of 
	the lattice $\mathcal{L}$. 
	All the bases of $\mathcal{L}$ have the same number of elements, i.e. in our case $n,$ which is called
	{\em dimension} or {\em rank} of $\mathcal{L}$. If $n=m$, then	the lattice $\mathcal{L}$ is said to have {\em full rank}. 
	We let $B$ be the $m\times n$ matrix, having as columns the vectors 
	${\bf{b}}_1,\ldots,{\bf{b}}_n$. 
	If $\mathcal{L}$ has full rank, then the {\em volume} of the lattice
	$\mathcal{L}$ is defined to be the positive number
	$|\det{B}|.$  The volume, as well as the rank, are independent of the basis $\mathcal{B}$. The volume is denoted 
	by $vol(\mathcal{L})$ or $\det{\mathcal{L}}.$ 
	Let now ${\bf v}\in \mathbb{R}^m$, then $\|{\bf v}\|$ denotes the Euclidean norm (or $
    \ell_2$ norm) of ${\bf v}$.  Additionally,  with  $\lambda_1(\mathcal{L})$ we write the least of the lengths of vectors in 	$ \mathcal{L}-\{ {\bf 0} \}$. If ${\bf t}\in {\rm{span}}({\bf b}_1,...,{\bf b}_n)$, then by $dist(\mathcal{L},{\bf t}),$ we mean $\min\{\|{\bf v}-{\bf t}\|: {\bf v}\in \mathcal{L} \}$. Finally, $P(B)=P({\mathcal{B}})=\{B{\bf x}:0 \le x_i < 1\ \text{for} \ i=1,2,...,n\}$ is the fundamental parallelepiped of ${\mathcal{B}}.$

Let $B_{m}({\bf r}_0,\beta)=\{{\bf x}\in{\mathbb{R}}^m : \|{\bf x}-{\bf r}_0\|<\beta\}$ the open ball of ${\mathbb{R}}^m$ with center ${\bf r}_0$ and radius $\beta.$  We denote the set $B_{m}({\bf 0},\beta)$ as  $B_{m}(\beta).$  The closed ball is ${\overline{B}}_{m}({\bf r}_0,\beta).$ We now provide the definition of the $i$th successive minimum.
\begin{definition}($i$-th successive minimum). Let ${\mathcal{L}}$ be a lattice of rank $n,$ then we define
\[
\lambda_i ({\mathcal{L}}) = \min \left\{ r : \dim \left( \operatorname{span} \left( {\mathcal{L}} \cap \overline{B}_n(r) \right) \right) \geq i \right\},\ 1\le i\le n.
\]
\end{definition}

We continue with the definition of SIVP.
\begin{definition}($SIVP_{\gamma}$). Let $\gamma=\gamma(n).$ Given a basis of ${\mathcal{L}}$ find a set of $n$ linearly independent lattice vectors $S = \{ \mathbf{s}_1,...,{\bf s}_n \} \subseteq {\mathcal{L}},$ such that  
\[\max_{1\le i\le n} \|\mathbf{s}_i\| \leq \gamma(n) \lambda_n({\mathcal{L}})\ \ (\gamma(n)\ge 1).\]
\end{definition}

We recall the definition of LLL reduction of a lattice basis. The LLL algorithm produces a reduced basis consisting of shorter vectors that are more orthogonal. For more details about LLL and general for background on lattices see \cite{galbraith}.

\begin{definition}(LLL)
A basis $\mathcal{B} = \{{\bf{b}}_1,\ldots,{\bf{b}}_n\}$ of a lattice $\mathcal{L}$ is called LLL-reduced if it satisfies the following conditions:
\\
\texttt{1.} $|\mu_{i,j}| = \frac{|{\bf b}_i \cdot {\bf b}^*_j|}{\|{\bf b}^*_j\|^2} \le \frac{1}{2} $
for every $i,j$ with $1\leq j < i \leq n$ (${\bf b}_j^*$ is the $j$-th Gram-Schmidt vector),
\\
\texttt{2.} $\|{\bf b}^*_i\|^2  \geq (\frac{3}{4}- \mu_{i,i-1}^2)\|{\bf b}^*_{i-1}\|^2$ for every $i$   with $1 < i \leq n$. 
\end{definition}

\begin{lemma}\label{Lemma:lll-vectors}
Let $\mathcal{B}=\{{\bf b}_1,...,{\bf b}_n\}$ be an LLL reduced basis with $\delta=3/4$ for a lattice ${\mathcal{L}} \subset \mathbb{R}^m$. Then, 
$$ \| {\bf b}_j\|\le 2^{(n-1)/2}\lambda_i({\mathcal{L}}), \text{ for } 1 \leq j \leq i \leq n.$$ 
\end{lemma}
\begin{proof}
See \cite[Theorem 17.2.12]{galbraith}.
\end{proof}

\subsection{Definition of ${\rm{SIS}}_{\mathbb{Z}}$} We now provide the definition of $SIS$ over the integers.
\begin{definition}
The \emph{$\ell_{\infty}$-SIS} problem over the integers, denoted 
$\ell_{\infty}\text{-}\mathrm{SIS}_{\mathbb{Z}}(n,m,\beta,{\mathcal{S}})$ with $m>n$, ${\mathcal{S}}\subset {\mathbb{Z}},$
is defined as follows.  
Let $A = (a_{ij}) \in \mathbb{Z}^{n \times m}$ be a random integer matrix whose entries $a_{ij} \in {\mathcal{S}}$ (${\mathcal{S}}$ is finite),  
and let $\beta > 0$ be given.  
The goal is to find a nonzero vector ${\bf z} \in \mathbb{Z}^m \setminus \{{\bf 0}\}$ such that,  
\[
A{\bf z} = {\bf 0}
\quad \text{and} \quad
\|{\bf z}\|_{\infty} \le \beta.
\]
\end{definition}
Finally, let $M$ be a positive real number, with $U_M$ we denote the set $\{  x\in {\mathbb{Z}}:|x|\le M\}=[-M,M]\cap {\mathbb{Z}}.$ We also use the non-symmetric set ${\mathcal{C}}_Q=\{0,1,...,Q-1\}$ for some positive integer $Q.$
\subsection{$SIS_{\mathbb{Z}}$ and $SIS_q$}\label{no-black-box-reuse}
\noindent 
In this subsection we explain why the standard worst-to-average reduction for
$\mathrm{SIS}_q$ cannot be reused as a black box for $\mathrm{SIS}_{\mathbb Z}$.
\medskip
\noindent
The implication
\[
\mathrm{SIS}_{\mathbb Z}\Longrightarrow \mathrm{SIS}_q
\]
is immediate: if $A{\bf z}={\bf 0}$ over $\mathbb Z$, then
$A{\bf z}\equiv {\bf 0}\pmod q$.

 The converse, however, is not straightforward.  In order to transfer a modular solution
to an exact integer solution, one needs a \emph{lifting property}. At the same time,
to invoke an $\mathrm{SIS}_q$ oracle in the usual way, one needs a \emph{uniformity
property} modulo $q$. We show that these two requirements are incompatible in our
parameter regime.

\begin{definition}[Lifting property $(LP)$]
Let $A\in {\mathcal{C}}_Q^{n\times m}$ and let $\beta,q\in\mathbb Z_{>0}$. We say that $(LP)$ holds
(for $(A,\beta,q)$) if every vector ${\bf z}\in\mathbb Z^m$ with
$\|{\bf z}\|_\infty\le \beta$ satisfying
\[
A{\bf z}\equiv {\bf 0}\pmod q
\]
also satisfies
\[
A{\bf z}={\bf 0}\qquad\text{over }\mathbb Z.
\]
\end{definition}

\begin{lemma}[Sufficient condition for $(LP)$]\label{lem:LP-sufficient}
Let $A=(a_{ij})\in {\mathcal{C}}_Q^{n\times m}$ and ${\bf z}\in\mathbb Z^m$ with
$\|{\bf z}\|_\infty\le \beta$. If
\[
q>m(Q-1)\beta,
\]
then
\[
A{\bf z}\equiv{\bf 0}\pmod q \quad\Longrightarrow\quad A{\bf z}={\bf 0}.
\]
Hence $(LP)$ holds whenever $q>m(Q-1)\beta$.
\end{lemma}

\begin{proof}
Write $A{\bf z}=(y_1,\dots,y_n)^T$. For each $i\in \{1,2,...,n \}$,
\[
|y_i|
=\Big|\sum_{j=1}^m a_{ij}z_j\Big|
\le \sum_{j=1}^m a_{ij}\,|z_j|
\le m(Q-1)\beta.
\]
If $A{\bf z}\equiv {\bf 0}\pmod q$ and $q>m(Q-1)\beta$, then each $y_i$ is a multiple of $q$
with $|y_i|<q$, hence $y_i=0$. Therefore $A{\bf z}={\bf 0}$.
\end{proof}
Now we provide the definition for the other property concerning the uniformity ${\rm{mod}} q$ of a matrix with coefficients in ${\mathcal{C}}_Q.$ 
\begin{definition}[Uniformity property $(UP)_\varepsilon$]
Let $A\xleftarrow{\$}{\mathcal{C}}_Q^{n\times m}$ be entrywise uniform and let $q\ge 2$.
For $\varepsilon>0$, we say that $(UP)_\varepsilon$ holds if the statistical distance (see definition \ref{def:statistical_distance})
\[
\Delta\!\bigl(A \bmod q,\ \mathcal U(\mathbb Z_q^{n\times m})\bigr)\le \varepsilon.
\]
\end{definition}
We first need the following auxiliary result. 
\begin{lemma}[Uniformity modulo $q$ for bounded integers]\label{lem:distribution-modq-1d}
Let $X\xleftarrow{\$}{\mathcal{C}}_Q$ and let $Y=X\bmod q\in\mathbb Z_q$, where $q\ge 2$.
\begin{enumerate}
\item $Y$ is uniform on $\mathbb Z_q$ if and only if $q\mid Q$.
\item If $q\nmid Q$, then
\[
0<\Delta\!\bigl(Y,\mathcal U(\mathbb Z_q)\bigr)
\le \frac{q}{4Q}.
\]
\end{enumerate}
\end{lemma}

\begin{proof}
For a proof of this result see Appendix \ref{sisz_sisq}.
\end{proof}

\begin{corollary}[Matrix version]\label{cor:distribution-modq-matrix}
Let $A\xleftarrow{\$}{\mathcal{C}}_Q^{n\times m}$ be entrywise uniform and let $q\ge 2$.
If $q\nmid Q$, then
\[
0<\Delta\!\bigl(A\bmod q,\ \mathcal U(\mathbb Z_q^{n\times m})\bigr)
\le \min\!\left\{1,\frac{nm\,q}{4Q}\right\}.
\]
\end{corollary}

\begin{proof}
Apply Lemma~\ref{lem:distribution-modq-1d} independently to each entry and use the product
bound for statistical distance (Lemma~\ref{Lemma:properties_of_stat_dist}(i)):
\[
\Delta\!\bigl(A\bmod q,\mathcal U(\mathbb Z_q^{n\times m})\bigr)
\le nm\cdot \Delta\!\bigl(Y,\mathcal U(\mathbb Z_q)\bigr)
\le \frac{nm\,q}{4Q}.
\]
The result follows.
\end{proof}
\begin{proposition}[Incompatibility of $(LP)$ and $(UP)_\varepsilon$]\label{prop:LP-UP-incompatible}
Assume $n,m>2$, $\beta\ge 1$, and $A\xleftarrow{\$}{\mathcal{C}}_Q^{n\times m}$.
If $(LP)$ is required for all vectors ${\bf z}$ with $\|{\bf z}\|_\infty\le \beta$, then a
sufficient condition is
\[
q>m(Q-1)\beta.
\]
On the other hand, if $(UP)_\varepsilon$ is required with small $\varepsilon$, then
Corollary~\ref{cor:distribution-modq-matrix} forces
\[
q \ll \frac{Q}{nm}\qquad\text{whenever } q\nmid Q.
\]
These two constraints are incompatible in the natural parameter regime.
In particular, one cannot simultaneously have both
\[
q>m(Q-1)\beta
\qquad\text{and}\qquad
\Delta\!\bigl(A\bmod q,\mathcal U(\mathbb Z_q^{n\times m})\bigr)\le \varepsilon
\]
with negligible (or even small constant) $\varepsilon$.
\end{proposition}

\begin{proof}
The lifting condition follows from Lemma~\ref{lem:LP-sufficient}.

For the uniformity condition, we distinguish two cases.

If $q\nmid Q$, then Corollary~\ref{cor:distribution-modq-matrix} gives
\[
\Delta\!\bigl(A\bmod q,\mathcal U(\mathbb Z_q^{n\times m})\bigr)
\le \frac{nm\,q}{4Q}.
\]
Thus, in order for the statistical distance to be small, one must require
$q=O(Q/(nm))$. Combining this with $q>m(Q-1)\beta$ yields (up to constants)
\[
m(Q-1)\beta \ll \frac{Q}{nm},
\]
hence
\[
nm^2\beta \ll 1,
\]
which contradicts $n,m>2$ and $\beta\ge 1$.

If $q\mid Q$, then necessarily $q\le Q$. In this case exact uniformity modulo $q$ may hold,
but the lifting condition still requires
\[
q>m(Q-1)\beta.
\]
Since $m\ge 3$ and $\beta\ge 1$, we get
\[
m(Q-1)\beta \ge 3(Q-1).
\]
Hence, for $Q\ge 2$,
\[
q>m(Q-1)\beta \ge 3(Q-1)\ge Q,
\]
which contradicts $q\le Q$.
\end{proof}

\begin{remark}
Proposition~\ref{prop:LP-UP-incompatible} shows that one cannot directly reuse the standard
$\mathrm{SIS}_q$ reduction by simply replacing the $\mathrm{SIS}_q$ oracle with an
$\mathrm{SIS}_{\mathbb Z}$ oracle. This is because the modulus $q$ must be large
to guarantee lifting, but small relative to $Q$ to preserve
uniformity modulo $q$. Therefore, a new worst-to-average reduction is needed for
$\mathrm{SIS}_{\mathbb Z}$.
\end{remark}

\section{Some basic lemmata}\label{sec:basic-lemmata}
In this section, we present some basic results related to lattices and linear systems. These include the smoothing parameter, Gaussian distribution and the Siegel's constant. 

Consider a basis ${\mathcal{B}}$ of a lattice and $B$ the corresponding matrix. We divide each basis vector into $Q$ equal segments, i.e., we partition the fundamental parallelepiped $P(B)$ into $Q^n$ cells. We then consider an
arbitrary $\mathbf{y} \in P(B)$ and we construct another point $\mathbf{z}$ that belongs to the same cell as $\mathbf{y}.$ We are interested in finding an upper bound for $\| \mathbf{y} - \mathbf{z} \|,$ which will be needed in the reduction in Section \ref{Sec:reduction}. For this reason, we begin with two useful definitions and a Lemma that will facilitate this bound.

\begin{definition}

The Frobenius norm of a matrix \( A \in \mathbb{C}^{n \times m} \) is defined as:

\[
\|A\|_F = \sqrt{ \sum_{i=1}^{n} \sum_{j=1}^{m} |a_{ij}|^2 } = \sqrt{ \text{tr}(A^* A) },
\]
where $A^*$ denotes the conjugate transpose of $A.$
\end{definition}
\begin{definition}

The spectral norm of a matrix \( A \in \mathbb{C}^{n \times m} \) is defined as:

\[
\|A\|_2 = \sup_{\|{\bf x}\| = 1} \|A {\bf x}\| = \max_{{\bf x} \ne {\bf 0}} \frac{\|A {\bf x}\|}{\|{\bf x}\|}.
\]
\end{definition}

\begin{lemma}\label{Lemma:||AX||}
$\| A{\bf x}\| \le \|A\|_2\|{\bf x}\|\le \|A\|_F\| {\bf x} \|.$
\end{lemma}
\begin{proof}
It is straightforward since $\| \cdot \|_2$ is submultiplicative and $\|A\|_2\le \|A\|_F.$ 
\end{proof}
\begin{lemma}\label{lemma:delta}
Let $B\in {\mathbb{Z}}^{n\times n}$ be a lattice basis matrix and $Q\in {\mathbb{N}}.$ Also, let ${\bf y}=B{\bf Y}\in P(B)$, for some ${\bf Y}\in [0,1)^n$. Then,\\
$({\rm{i}}).$ If ${\bf a}=\lfloor QB^{-1}{\bf y}\rfloor,$ ${\bf t}=\frac{1}{Q}{\bf{a}},$ ${\boldsymbol{\delta}}={\bf Y}-{\bf t},$ then ${\boldsymbol{\delta}}\in [0,1/Q)^n.$ \\
$({\rm{ii}}).$  Let $B\in U_M^{n\times n}$ and ${\bf z}=B{\bf t}$. Then, 
    \[\|{\bf y}-{\bf z}\|\leq \frac{n\sqrt{n}M}{Q}.\]
\end{lemma}
\begin{proof}
$({\rm{i}}).$ Let ${\bf Y}=(Y_i)_i$, ${\bf a
}=(a_i)_i,$ and ${\bf t}=(t_i)_i.$ Then,
    $${\boldsymbol{\delta}}={\bf Y}-{\bf t}=(Y_1-t_1,...,Y_n-t_n)=
    \big(Y_1-\frac{1}{Q}a_1,...,Y_n-\frac{1}{Q}a_n\big).$$
    Since 
    $${\bf a}=\lfloor QB^{-1}{\bf y}\rfloor=\lfloor Q{\bf Y}\rfloor=(\lfloor QY_1\rfloor,...,\lfloor QY_n \rfloor),$$ 
    we get 
    $$\delta_i=Y_i-\frac{1}{Q}a_i=Y_i-\frac{1}{Q}\lfloor QY_i\rfloor\in [0,1/Q)\ (1\le i\le n).$$ 
    So $\boldsymbol{\delta}\in [0,1/Q)^n.$ We used that for every real $x$ we have $x-\frac{\lfloor Qx \rfloor}{Q}=\frac{\{Qx\}}{Q}\in [0,1/Q),$ where $\{Qx\}$ is the fractional part of $Qx.$\\
$({\rm{ii}}).$   \[\|{\bf y}-{\bf z}\| \leq \| B\boldsymbol{\delta}\|\le \| B\|_2 \| \boldsymbol{\delta}\|\le \|B\|_F\| {\bf \boldsymbol{\delta}}\|.\]
    But $\|B\|_F= \sqrt{ \sum_{i=1}^{n} \sum_{j=1}^{n} |b_{ij}|^2 }\le  \sqrt{ \sum_{i=1}^{n} \sum_{j=1}^{n} M^2}=\sqrt{n^2M^2}=nM $ and since $\delta\in [0,1/Q)^n$ we have  $\|\boldsymbol{\delta}\|\le \sqrt{n}/Q.$ So,
    \[\|{\bf y}-{\bf z}\| \le \frac{n\sqrt{n}M}{Q}.\]
\end{proof}

\subsection{Gaussian distribution}
Let $\rho_{\sigma,{\bf c}}({\bf x}) =e^{-\pi\|({\bf x}-{\bf c})/\sigma\|^2}$  $({\bf x}\in {\mathbb{R}}^n)$ be the Gaussian  function with center ${\bf c}$ scaled by factor $\sigma.$ We define the Gaussian distribution 
$D_{\sigma,{\bf c}}({\bf x})=\frac{\rho_{\sigma,{\bf c}}({\bf x})}{\sigma^n},$ since $\int_{{\mathbb{R}^n}}\rho_{\sigma,{\bf c}}({\bf x})\ d{\bf x}=\sigma^n.$
If ${\bf c}={\bf 0}$ we write $D_{\sigma}.$ 
Now, $D_{\sigma}={\mathcal{N}}({\bf 0},{\sigma'}^2I_n),$ with $\sigma' = \frac{\sigma}{\sqrt{2\pi}}.$ 
and the normal distribution ${\mathcal{N}}({\bf 0},{\sigma'}^2I_n),$ is defined as follows (we consider ${\bf c}={\bf 0}$), 
$$N_{\sigma'}({\bf x})=\frac{1}{{(2\pi\sigma'^2)}^{n/2}}\exp{\Big(-\frac{\|{\bf x}\|^2}{2\sigma'^2}\Big)}.$$

With $a<_p b$ ($a$ follows some distribution and $b$ is a constant) we mean that the inequality holds with probability $p.$

\begin{lemma}\label{Lemma:length_of_gaussian}
Let ${\bf x}\in {\mathbb{R}}^n$, such that ${\bf x}\xleftarrow{\$}D_{\sigma}.$ Then, $ \| \mathbf{x} \| \leq_p \sqrt{2n} \sigma$ for $p > 1 - 2^{-\Omega(n)}.$
\end{lemma}
\begin{proof}
Result follows by applying \cite[Theorem 3.1.1]{vershynin} with variance $\sigma'^2.$
\end{proof}

\begin{lemma}[Linear Combination of Gaussians]\label{lemma:linear_combination}
Let $\mathbf{x}_1,\ldots,\mathbf{x}_m \in \mathbb{R}^n$ be independent random vectors drawn from the Gaussian distribution $D_\sigma$ over $\mathbb{R}^n$. Let ${\bf a}=(a_1,\ldots,a_m) \in \mathbb{R}^m$ and define
\[
\mathbf{S} = \sum_{j=1}^m a_j \mathbf{x}_j.
\]
Then, with probability at least $1 - 2^{-\Omega(n)}$, it holds that
\[
\|\mathbf{S}\| \le \sqrt{2n}\,\sigma\,\|{\bf a}\|.
\]
\end{lemma}

\begin{proof}
Write $s=\sigma/\sqrt{2\pi}$ so that $\mathbf{x}_j\sim{\mathcal N}({\bf 0},s^2I_n)$.
By independence and linearity, $\mathbf{S}\sim{\mathcal N}({\bf 0},s^2\|{\bf a}\|^2 I_n)$.
Translating back to the $D_\sigma$ convention gives $\mathbf{S}\sim D_{\sigma\|{\bf a}\|}$.
As in the previous Lemma we get:
\[
\|\mathbf{S}\| \le_p \sqrt{2n}\, \sigma \|{\bf a}\|.
\]
\end{proof}
\begin{corollary}\label{Corollary:tail_bound}
    If $\max_{1\le i\le m}\{|a_i|\}\leq\beta$ we get
    $$\|{\bf S}\|\le_p \sigma \beta\sqrt{2nm}.$$
\end{corollary}
This is immediate since $\|{\bf a}\|\le \beta \sqrt{m}.$
\subsection{Discrete Gaussian}

We define the \emph{discrete Gaussian distribution} with center ${\bf c}$ as
\begin{equation}
    D_{\mathcal{L}, \sigma, {\bf c}}({\bf{x}}) = \frac{\rho_{\sigma,{\bf c}}({\bf{x}})}{\rho_{\sigma,{\bf c}}(\mathcal{L})}, \ {\bf x}\in{\mathcal{L}},
\end{equation}
where 
\[
 \rho_{\sigma,{\bf c}}(\mathcal{L}) = \sum_{{\bf x} \in \mathcal{L}} \rho_{\sigma,{\bf c}}({\bf x}).
\]
When ${\bf c}={\bf 0}$ we write $D_{\mathcal{L}, \sigma}({\bf x})$ and $\rho_{\sigma}(\mathcal{L}).$

\subsection{Smoothing parameter}

We start with the definition of dual lattice ${\mathcal{L}}^*.$
\begin{definition}
    \[
{\mathcal{L}}^* = \left\{ \mathbf{y} \in \mathbb{R}^n \mid \langle \mathbf{x}, \mathbf{y} \rangle \in \mathbb{Z}, \, \forall \, \mathbf{x} \in {\mathcal{L}} \right\}.
\]
\end{definition}

In \cite{Mic-Regev}, Micciancio and Regev introduced a lattice parameter known as {\it{smoothing parameter}} of a lattice. This is used to make the discrete Gaussian statistically close to uniform modulo ${\mathcal{L}}$ up to an error $\varepsilon$.
\begin{definition}[Smoothing parameter]
    For an \( n \)-dimensional lattice \({\mathcal{L}}\) and a positive real number \( \varepsilon > 0 \), the smoothing parameter \( \eta_\varepsilon({\mathcal{L}}) \) is the smallest \( \sigma > 0 \) such that the Gaussian mass of the dual lattice \( {\mathcal{L}}^* \) (excluding the origin) is bounded by \( \varepsilon \). Formally, it is defined as:
\[
\eta_\varepsilon({\mathcal{L}}) = \min \left\{ \sigma > 0 : \rho_{1/\sigma}({\mathcal{L}}^* \setminus \{{\bf{0}}\}) \leq \varepsilon \right\},
\]
where:
\[
\rho_\sigma(\mathbf{x}) = e^{-\pi \|\mathbf{x}\|^2 / \sigma^2}
\]
is the Gaussian function with parameter \( \sigma \), and
\[
\rho_\sigma({\mathcal{L}}^* \setminus \{{\bf{0}}\}) = \sum_{\mathbf{y} \in {\mathcal{L}}^* \setminus \{{\bf{0}}\}} e^{-\pi \|\mathbf{y}\|^2 / \sigma^2}.\]
\end{definition}

We shall use the smoothing parameter in order to produce random points in the fundamental parallelepiped of the lattice for a specific basis. 
\begin{lemma}\label{Lemma:properties_of_stat_dist}
    $({\bf i}).$ Let ${\bf X}=(X_1,...,X_m), {\bf Y}=(Y_1,...,Y_m)$ be two lists of independent random variables, then,
    $\Delta({\bf X},{\bf Y})\le \sum_{i=1}^{m}\Delta(X_i,Y_i).$\\
    $({\bf ii}).$ Let $X,Y$ be two random variables over a common set $A.$ For any 
    function $f$ with domain $A$ we have $\Delta(f(X),f(Y))\le \Delta(X,Y).$
    \end{lemma}
\begin{proof}
    For $({\bf i})$ see \cite[Proposition 8.9]{book-Mic-Gold}. For  $({\bf ii})$
see \cite[Proposition 8.10]{book-Mic-Gold}
\end{proof}
We have the following basic result.
\begin{proposition}\label{Prop:smooth_parameter}
For any  $\sigma > 0,\, {\bf c} \in \mathbb{R}^n,$ and lattice ${\mathcal{L}}={\mathcal{L}}({{B}}),$ the statistical distance between $D_{\sigma,{\bf{c}}} \mod P({{B}})$ and the uniform distribution over $P({{B}})$ is at most $\frac{1}{2} \rho_{1/\sigma}({\mathcal{L}}^* \setminus \{{\bf{0}}\}).$
In particular, for any $\varepsilon > 0$ and any $\sigma \geq \eta_{\varepsilon}({\mathcal{L}}),$ we get
$$\ \Delta\big(D_{\sigma,{\bf{c}}}\ {\rm{mod}}\  P({{B}}),\ {\mathcal{U}}(P({{B}}))\big) \leq \frac{\varepsilon}{2}.$$
\end{proposition}
\begin{proof}
    \cite[Lemma 4.1]{Mic-Regev}
\end{proof}

That is,  if we sample points from a Gaussian distribution with parameter $\sigma\geq \eta_{\varepsilon}({\mathcal{L}})$ and reduce them to the fundamental  parallelepiped  of a lattice, the distribution of those points is nearly the same as if the points were sampled uniformly from the fundamental parallelepiped. 

Furthermore, we need the following well known auxiliary results for bounding the smoothing parameter.
\begin{lemma}\label{Lemma:smoothing-bound}
For any \( n \)-dimensional lattice \( {\mathcal{L}} \) and a positive real \( \varepsilon > 0 \), we have
\[
\eta_\varepsilon({\mathcal{L}}) < {\sqrt{\ln\big(2n(1 + 1/\varepsilon)\big)}}  \lambda_n({\mathcal{L}}).
\]
\end{lemma}

\begin{proof}
See \cite[Lemma 3.3]{Mic-Regev}    
\end{proof}
We have the following,
\begin{corollary}\label{cor:smoothing-bound}
    For $\varepsilon=n^{-\log_2{n}},$ we get 
    $\eta_\varepsilon({\mathcal{L}}) < 2\lambda_n({\mathcal{L}})\log_2{n}.$
\end{corollary}
\begin{proof}
    See \cite[Lemma 15]{Regev}
\end{proof}
Also we have the following result.
\begin{lemma}\label{Lemma:smoothing-bound2}
    For $\varepsilon<1/100,$ we get
    $\eta_{\varepsilon}({\mathcal{L}})>\frac{1}{n}\lambda_n({\mathcal{L}}).$
\end{lemma}
\begin{proof}
See appendix \ref{upper_bound_for_smooth+parameter} for a proof.
\end{proof}
Finally, we shall need the following Lemma concerning the distribution $D_{{\mathcal{L}+{\bf y}},\sigma}$ defined by 
\[Pr\big({\bf x}\xleftarrow{\$} D_{{\mathcal{L}+{\bf y}},\sigma}\big)=\frac{\rho_{\sigma}({\bf{x}})}{\rho_{\sigma}(\mathcal{L}+{\bf y})}, \ {\bf x}\in{\mathcal{L}}+{\bf y}.\]
\begin{lemma}\label{Lemma:full_dimensional}
Let ${\mathcal{L}}$ be a lattice of dimension $n.$ For $\sigma \geq \sqrt{2} \eta_\varepsilon({\mathcal{L}}),$ $\varepsilon=n^{-\log_2{n}},$ 
and any $(n - 1)$-dimensional hyperplane $H,$ we get
\[ \Pr({\bf x} \xleftarrow{\$} D_{{\mathcal{L}+{\bf y}},\sigma}: {\bf x} \in H) < 0.9.\]
\end{lemma}
\begin{proof}
\cite[Lemma 14]{Regev}. 
\end{proof}

\begin{lemma}\label{lemma:indepedent_vectors}
Let ${\mathcal A}$ be an algorithm which, on input a full-rank lattice 
${\mathcal L} \subset {\mathbb R}^n$, outputs a non-zero vector 
${\bf w}$ of the form ${\bf x} - {\bf y}$ (with  probability at least $1/n^{c_0}$), where 
${\bf x} \sim D_{{\mathcal L}+{\bf y},\sigma}$ for some ${\bf y}\in {\rm{span}({\mathcal{L}})}$, 
and $\sigma \ge \sqrt{2}\,\eta_{\varepsilon}({\mathcal L})$. Then, after at most $O(n^{2+c_0})$ calls to ${\mathcal A}$, we obtain $n$ independent vectors of 
${\mathcal L}$ with overwhelming probability.
\end{lemma}
\begin{proof}
First remark that ${\bf w}\in {\mathcal{L}}.$
Fix a $k\in {\mathbb{Z}}_{\geq 2}.$ Say, after some calls to algorithm ${
\mathcal{A}}$ we have collected vectors 
${\bf w}_1,\ldots,{\bf w}_{k-1} \in {\mathcal L}$, and let
\[
U_{k-1} = {\rm{span}}( {\bf w}_1,\ldots,{\bf w}_{k-1} )
\qquad
{\rm{is\ such \ that\ }} d_{k-1} = \dim(U_{k-1}) < n.
\]
On the next successful call we obtain ${\bf w}_k \in {\mathcal L}$. $U_{k-1}$ is contained in some hyperplane say $H_{k-1},$ so 
$\Pr[\,{\bf w}_k \in U_{k-1}\,]\le \Pr[\,{\bf w}_k \in H_{k-1}\,].$
Since $\sigma \ge \sqrt{2}\,\eta_{\varepsilon}({\mathcal L})$, by Lemma~\ref{Lemma:full_dimensional}\footnote{Lemma refers to ${\bf x}$ that follows $D_{{\mathcal{L}}+{\bf y},\sigma}.$ However ${\bf w}$ follows same type of distribution with the same variance but different center ${\bf y}'$. In fact, ${\bf w}$ follows $D_{{\mathcal{L}},\sigma,-{\bf y}}$.} we have,
\[
\Pr[\,{\bf w}_k \in H_{k-1}\,] \le 0.9,
\quad\text{and hence}\quad
\Pr[\,{\bf w}_k \in U_{k-1}\,] \le 0.9.
\]
Therefore,
\[\Pr[\,{\bf w}_k \notin U_{k-1}\,] \ge 0.1.\]
Let $X$ be the random variable denoting the number of outputs 
from ${\mathcal A}$ required to obtain $n$ independent vectors. 
Since each step increases the dimension with probability $p$ at least 
$0.1$, the random variable $X$ is a 
negative binomial distribution with parameters $(n,p)$. 
Therefore,
\[
\mathbb{E}[X] = \frac{n}{p} \le \frac{n}{0.1} = 10n.
\]
Hence, the expected number of successful outputs needed to obtain 
$n$ linearly independent lattice vectors is at most say $O(n^2).$

To prove this holds with overwhelming probability, let $N=20n$ be the number of successful outputs collected. Also, let $X_i$ be an indicator variable where $X_i = 1$ if the $i$-th vector 
increases the dimension of the span, and $X_i = 0$ otherwise. 
The dimension of the final span is $D = \sum_{i=1}^{N} X_i$.

But we have ${\mathbb{E}}[X_i] \ge  0.1$. 
We safely have ${\mathbb{E}}[D] \ge 2n$. 
Using the Chernoff bound (second inequality below), we get:
\[
\Pr[\, D < n \,] \le \Pr[\, D < (1 - 1/2){\mathbb{E}}[D] \,] \le e^{-\frac{(1/2)^2 \cdot 2n}{2}} = e^{-n/4}.
\]
Since the algorithm $\mathcal{A}$ 
succeeds with probability at least $1/n^{c_0}$, we only need at most $O(n^2 \cdot n^{c_0})$ total 
calls to ensure we collect $n$ successful vectors with high probability. 
By a union bound, the total failure probability remains negligible.
\end{proof}

\subsection{Siegel's constant}
In the present work we study linear systems over a finite subset of integers. So, we shall need a sufficient condition to ensure the existence of such solutions. We start with Siegel's Lemma.
\begin{theorem}(Siegel's Lemma).\label{Theorem:Siegel-lemma}
Let $U_{M}=\{  x\in {\mathbb{Z}}:|x|\le M\}$ and $A\in U_M^{n\times m},$ where $m>n.$ Then, the linear system $AX={\bf 0}$ has a non-zero integer solution with 
\[\max_{1\le i\le m}\{ |x_i|\}\le \lfloor (mM)^{n/(m-n)}\rfloor.\]   
\end{theorem}
\begin{proof}
For a proof see \cite{siegel}.
\end{proof}
We call the quantity $c_S(A,M)=(mM)^{n/(m-n)}$ Siegel's constant corresponding to the matrix $A\in U_M^{n\times m}$.  Say, we want to have a solution with \[\max_{1\le i\le m}\{ |x_i|\}\le \beta,\] then it is enough to choose
$\beta > c_S(A,M).$

Now, we consider the following problem. 
We choose a non-singular matrix $B \in \mathbb{Z}^{n \times n}$ with entries bounded by $M$, and define a new $n\times m$ matrix, 
\[
A = [\,{\bf a}_1, \ldots, {\bf a}_m\,], m=m(n),
\]
where each column is given by 
\[
{\bf a}_j = \lfloor Q B^{-1} {\bf y}_j \rfloor,
\]
for some integer $Q$ and vectors ${\bf y}_j.$
We are interested in determining conditions involving $n$, under which the system 
\[
A{\bf X}= {\bf 0}
\]
admits a short nontrivial solution. 

\begin{proposition}\label{Prop:siegel_constant}
    Let $B$ be a matrix from $U_M^{n\times n}$ and ${\mathcal{L}}(B)$ the lattice generated by the columns of $B$. Set   $m = (n + 1)n,$
    and consider the matrix $A$ $\in \mathbb{Z}^{n \times m}$ with columns given by the vectors 
    \[
    {\bf a}_j=\lfloor{QB^{-1}{\bf y}_j}\rfloor~ (j=1,\ldots,m),
    \]
    for some ${\bf y}_j\in P({B})$. Then, \\
    ${\rm (i)}.$ $A\in \mathcal{C}_Q^{n \times m}.$ \\
    ${\rm (ii)}.$ The Siegel's constant $c_S(A,Q)$ is bounded above by $\beta_n = (2n^2Q)^{1/n}.$ 
\end{proposition}
\begin{proof}
${\rm (i)}.$ Let ${\bf a}_j=(a_{j1},a_{j2},...,a_{jn}).$ Since $\mathbf{y}_j \in P(B)$, there is a $\mathbf{t}_j \in [0,1)^n,$ such that $\mathbf{y}_j = B\mathbf{t}_j.$ Therefore, ${a}_{j,i} = \lfloor Q{t}_{j,i}\rfloor,$ and
$0 \leq{a}_{j,i}  < Q,$
so $A\in  \mathcal{C}_Q^{n \times m}.$\\\\
${\rm (ii)}.$ Since $\frac{n}{m-n} = \frac{1}{n}$ and $m<2n^2$ we get,
\[
c_S(A,Q) = (mQ)^{n/(m-n)}
< ( 2 n^2 Q)^{1/n}.
\]
 If we set $\beta_n = (2Qn^2)^{1/n}$ the result follows.
\end{proof}
\begin{remark}
    Since $\beta_n=(2n^2Q)^{1/n}$ and $Q=\mathrm{poly}(n)$, we have $\beta_n\to 1$; hence for all $n\ge 5$ we may fix $\beta=2$ and therefore treat $\beta=O(1)$ throughout.
\end{remark}

\begin{corollary}\label{cor:beta}
   Let $\beta$ be a positive integer independent of $n$. Using the terminology of the previous Proposition, the linear system $A{\bf X}={\bf 0}$ has a solution with 
    $\max_{1\le i \le m} \{| {x}_i|\}\le \beta$  if $\beta_n < \beta.$
\end{corollary}
\section{Solving ${\rm{SIVP}}_{\tilde{O}(n^{1.5})}$ with an ${\rm{SIS}}_{\mathbb{Z}}$-oracle ${\mathcal{O}}$}\label{Sec:reduction}
In this section we build a {\it{Short Vectors}} algorithm (Algorithm \ref{alg:alg1}) that uses an ${\rm{SIS}}_{\mathbb{Z}}$ oracle in order to construct our reduction to SIVP, i.e. given an oracle that solves ${\rm{SIS}}_{\mathbb{Z}}$ we solve SIVP with approximation factor $\tilde{O}(n^{1.5})$ for any lattice.

Let ${\mathcal{L}}$ be a full rank lattice of dimension $n$ and $\eta_{\varepsilon}({\mathcal{L}})$ the smoothing parameter for 
the negligible function $\varepsilon=\frac{1}{n^{\log_2{n}}}.$ 
Let $\tilde{\eta}$ be a real number in the interval $[2\eta_{\varepsilon}({\mathcal{L}}),4\eta_{\varepsilon}({\mathcal{L}})].$ We consider a basis of ${\mathcal{L}},$ say $\mathcal{B}=\{{\bf b}_1,...,{\bf b}_n\},$ $B$ the matrix with columns the vectors of ${\mathcal{B}},$ and let $P({B})$ be the fundamental parallelepiped of $\mathcal{B}.$ 

We assume that we can solve ${\rm{SIS}}_{\mathbb{Z}}(n, m, \beta, {\mathcal{C}}_Q )$, so for $A$ randomly chosen from $ {\mathcal{C}}_Q^{n\times m}$, in time $O(poly(m,n))$ we can find an integer solution with the $\ell_{\infty}$-norm, at most $\beta.$ We assume that $m=poly(n)$, so there is some constant $c_0>0$ such that $O\big(poly(n,m(n))\big)=O(n^{c_0}).$

We shall use the following Algorithm \ref{alg:alg1} in order to construct the worst-to-average case reduction. In Theorem \ref{theorem:basic} we show the correctness of Algorithm \ref{alg:alg1}, and how to use it to construct our reduction.

\begin{algorithm}[H]
\caption{Short Vectors}\label{alg:alg1}
\begin{algorithmic}[1]
\REQUIRE Let $n,M$ positive integers. We consider a matrix $B \in U_M^{n\times n}$ with columns linearly independent. ${\mathcal{L}}$ is the lattice generated by the columns of $B.$ We set $m=(n+1)n$ and $\beta=O(1)$. Also, let $\tilde{\eta}$ be a real positive number in $[2\eta_{\varepsilon}({\mathcal{L}}),4\eta_{\varepsilon}({\mathcal{L}})].$.\\
\ENSURE A short non zero lattice vector ${\bf v}$ or fail. \\
\STATE Choose $m-$vectors ${\bf x}_i\xleftarrow{\$} {D_{\tilde{\eta}}}$  $(i=1,2,\dots,m).$
\STATE Set ${\bf y}_i$ such that ${\bf y}_i\equiv {\bf x}_i\pmod{P(B)}$ $(i=1,2,\dots,m).$
\STATE Set $Q=\lceil n\sqrt{m}M\rceil.$ Let $A$ be the $n\times m$ matrix with columns ${\bf a}_i = \lfloor QB^{-1}{\bf y}_i \rfloor$ $(i=1,2,\dots,m).$
\STATE With (non-negligible) probability at least $1/n^{c_0}$ we solve $A{\bf X}={\bf 0}$ and say ${\bf r}=(r_1,\dots,r_m)\in {\mathbb{Z}}^m$ be a solution with 
$\|{\bf r}\|_{\infty}\le \beta.$
\STATE If the previous step succeeds, the algorithm returns the lattice vector 
${\bf v}=\sum_{j=1}^{m}r_j({\bf y}_j-{\bf x}_j).$ Else, it returns fail.
\end{algorithmic}
\end{algorithm}

\begin{remark}
    In step 4 we assume the existence of an oracle, which on input a random integer matrix from $ {\mathcal{C}}_Q^{n\times m}$ for some $M$, outputs an integer solution with $\ell_{\infty}$-norm at most $\beta$ (with non-negligible probability). In order to use this oracle, we must show that $A$ constructed in step 3 is uniform. This is proved in Theorem \ref{theorem:basic}(iii).
\end{remark}


 \noindent \textbf{Algorithmic Differences: SIS$_q$ vs. SIS$_\mathbb{Z}$.} In the original reduction, as provided by Ajtai, the modulus $q$ is used to define the SIS$_q$ instance to induce collisions in $\mathbb{Z}^n_q$. In our algorithm, we define the parameter $Q$, whose only purpose is to help us bound the length of the output vector $\mathbf{v}$. That is, by choosing  $Q=poly(n)$, we get the bound $\tilde{O}(n^{1.5} \lambda_n(\mathcal{L})).$ These two parameters, $q,Q$ in ${\rm{SIS}}_q, {\rm{SIS}}_{\mathbb{Z}}$ (resp.), serve different roles and they should not be conflated.
In the original ${\rm{mod}} q$ reduction, the $z_i = \frac{1}{q}B\mathbf{a}_i$ vectors are computed as part of the algorithm and are used to determine the output vector. Our reduction algorithm does not require them. Finally, we use a different SIS oracle, which solves the SIS over $\mathbb{Z},$ not the SIS over $\mathbb{Z}_q.$
\ \\

Parts (i) and (iii) of the Theorem below, establish the correctness of Algorithm \ref{alg:alg1}. Part (ii) gives an upper bound on the Euclidean norm of the output lattice vector, and in  (iv) we show that with overwhelming probability after polynomially many calls to Algorithm \ref{alg:alg1} we get $n$ independent vectors of ${\mathcal{L}}.$
\begin{theorem}\label{theorem:basic}
$({\bf i}).$ If Algorithm \ref{alg:alg1} does not fail, the output ${\bf v}$ is a lattice vector.\\
$({\bf ii}).$ If the oracle ${\mathcal{O}}$, on input $A$ (from step 3 of Algorithm \ref{alg:alg1}), succeeds in ${\rm{SIS}}_{\mathbb{Z}}(n, m, \beta,  {\mathcal{C}}_Q ),$ we get  
$$\| {\bf v}\| \le_p \tilde{O}(n^{1.5} \lambda_n({\mathcal{L}})),$$ and the inequality holds with probability $p> 1 - 2^{-\Omega(n)}.$\\ 
$({\bf iii}).$ The matrix $A$, of step 3 of Algorithm \ref{alg:alg1}, is statistically close to the uniform distribution on
${\mathcal C}_Q^{\,n\times m}$, where ${\mathcal C}_Q=\{0,1,\dots,Q-1\}$. More precisely,
\[
\Delta\big(A,{\mathcal U}({\mathcal C}_Q^{\,n\times m})\big)\le \frac{m\varepsilon}{2}\ (\varepsilon=n^{-\log_2{n}}).
\]
In particular, for large enough $n
$, Algorithm \ref{alg:alg1} on input $A$ succeeds with probability  at least
$\frac{1}{2n^{c_0}}.$\\

$({\bf iv}).$ After at most $O(n^{c_0+2})$ calls to the Algorithm \ref{alg:alg1} (for fixed $\tilde{\eta}$) we get $n$ independent vectors of ${\mathcal{L}}$.
\end{theorem}  
\begin{proof}

$({\bf i}).$ By definition of reduction modulo $P(\mathcal B)$, since ${\bf y}_j\equiv {\bf x}_j \pmod{P(\mathcal B)}$,
there exist real numbers $\kappa_{ji}$ $(1\le j\le m, 1\le i\le n)$ such that 
\[ {\bf y}_j-{\bf x}_j=\sum_{i=1}^{n}{\bf b}_i(\kappa_{ji}-\lfloor{\kappa_{ji}}\rfloor)-\sum_{i=1}^{n}\kappa_{ji}{\bf b}_i=-\sum_{i=1}^{n}\lfloor \kappa_{ji}\rfloor{\bf b}_i\in {\mathcal{L}}
\]
for every $j=1,2,...,m$. So ${\bf v}\in {\mathcal{L}}.$\\
$({\bf ii}).$ We set ${\bf z}_j=\frac{1}{Q}B{\bf a}_j$, for $j\in\{1,2,...,m\},$ where with $B$ we denote the matrix with columns ${\bf b}_1,...,{\bf b}_n$. 
From step 4 of the Algorithm \ref{alg:alg1} we get a solution ${\bf r}=(r_1,...,r_m)$ of $A{\bf X}={\bf 0}$ with $\max\{|{r}_j|\}\leq\beta.$ 
Note that
\[ 
\sum_{j=1}^{m}r_j{\bf z}_j=\frac{1}{Q}\sum_{j=1}^{m}r_jB{\bf a}_j=\frac{1}{Q}B\sum_{j=1}^{m}r_j{\bf a}_j={\bf 0}.
\]
Then,
\[ {\bf v}=\sum_{j=1}^{m} r_j({\bf y}_j-{\bf x}_j) - \sum_{j=1}^{m}r_j{\bf z}_j=-\sum_{j=1}^{m}r_j{\bf x}_j+\sum_{j=1}^m r_j({\bf y}_j-{\bf z}_j).\]
Thus,
\begin{equation}\label{eq2}
\| {\bf v}\| \le  \big\|\sum_{j=1}^{m}r_j {\bf x}_j\big\| + \max_{1\le j\le m}|r_j|\sum_{j=1}^{m}\|{\bf y}_j-{\bf z}_j\|.  
\end{equation}
 From Corollary \ref{Corollary:tail_bound}, we get $\|\sum_{j=1}^{m} r_j {\bf x}_j\|\le_p \sqrt{2nm}\beta \tilde{\eta}$ for $p$ close to 1 (in fact $p>1-2^{-\Omega(n)})$. Furthermore, from Lemma \ref{lemma:delta} (ii) we have $\|{\bf y}_j-{\bf z}_j\|\le \frac{n\sqrt{n}M}{Q}.$
So relation (\ref{eq2}) gives,
\begin{equation}\label{equation:select_q}
    \| {\bf v}\| \le_p \beta\Big(\tilde\eta\sqrt{2nm} + \frac{Mmn\sqrt{n}}{Q}\Big).
\end{equation} 
Since,
$Q=\lceil n\sqrt{m}M\rceil\ge n\sqrt{m}M$, we get 
\[ 
\frac{mn\sqrt{n}M}{Q}\le \sqrt{nm}\leq \sqrt{nm}\lambda_n({\mathcal{L}}),
\]
and $\tilde{\eta}\le 4\eta_{\varepsilon}({\mathcal{L}})< 8(\log_2{n})\lambda_n({\mathcal{L}})$ (from Corollary \ref{cor:smoothing-bound}), thus
\begin{equation}\label{complexity_inequality}
    \| {\bf v}\| \le_p 16\sqrt{2}\beta \sqrt{nm} (\log_2{n}) \lambda_n({\mathcal{L}}).
\end{equation}
Finally, $m=O(n^2)$ and $\beta=O(1),$ so
\begin{equation}\label{equation_m}
    \| {\bf v}\| \le_p  \tilde{O}(n^{1.5} \lambda_n({\mathcal{L}})).
\end{equation}  

$({\bf iii}).$ We must prove that the input to the oracle is statistically close to uniform. I.e. we have to prove that ${\bf a}_i = \lfloor QB^{-1}{\bf y}_i \rfloor$ $(i=1,2,\dots,m)$ of step 3 of the Algorithm \ref{alg:alg1} follows a distribution that is statistically close to the uniform distribution.
Let $g({\bf y})=\lfloor QB^{-1}{\bf y}\rfloor,$ ${\bf y}\in P(B).$ Since ${\bf y} \in P(B)$, we have $B^{-1}{\bf y} \in [0,1)^n$. Therefore
\[
g({\bf y})=\lfloor Q B^{-1}{\bf y} \rfloor \in \{0,\ldots,Q-1\}^n.
\] 
We set ${\mathcal{C}}_Q=\{0,1,...,Q-1\}.$ So, $g:P(B)\rightarrow {\mathcal{C}}_Q^{n}.$
If $\mathbf y$ is uniformly distributed over $P(B)$, then the random variable
$g(\mathbf y)$ is uniformly distributed on ${\mathcal{C}}_Q^n.$ Indeed, since $Q\ge 2$ and integer we get $Pr_{{\bf y}\sim {\mathcal{U}(P(B))}}[g({\bf y})={\bf z}]=1/Q^n, $ for all ${\bf z}\in {\mathcal{C}}_Q^n.$  We denote by \(\mathcal{U}(W)\) the uniform distribution over a finite set \(W\). Thus,
$$g\big({\mathcal{U}}(P(B))\big)={\mathcal{U}}({\mathcal{C}}_Q^n).$$
By the data-processing property of statistical distance, replacing $\mathbf y$
by a distribution that is close to uniform over $P(B)$ changes the
distribution of $g(\mathbf y)$ by at most the same statistical distance.

Then, the statistical distance
\[\Delta\Big( \big(g({\bf y}_1),...,g({\bf y}_m)\big), \big({\mathcal{U}}({\mathcal{C}}_{Q}^n)\big)^{m}\Big)\underset{{\rm{Lemma\ }} \ref{Lemma:properties_of_stat_dist} ({\rm i})}{\leq} \]
\[\le \sum_{i=1}^{m}\Delta\big(g({\bf y}_i),{\mathcal{U}}({\mathcal{C}}_{Q}^n)\big)\underset{\rm{the\ previous\ remark}}{=} 
\sum_{i=1}^{m}\Delta\Big(g({\bf y}_i),g\big({\mathcal{U}}\big(P({{B}})\big)\big)\Big)\underset{{\rm{Lemma\ }} \ref{Lemma:properties_of_stat_dist} ({\rm ii})}{\leq} \] 
\[\le \sum_{i=1}^{m}\Delta\big({\bf y}_i,{\mathcal{U}}(P({{B}}))\big)=m\Delta\big(D_{\tilde{\eta}}{\ \rm{mod}} P({{B}}),{\mathcal{U}}(P({{B}}))\big)\underset{{\rm{Proposition}} \ \ref{Prop:smooth_parameter}}{\le} m\varepsilon/2.\]
We used the same notation $\Delta$ for the total variation distance
on arbitrary probability spaces, both finite and continuous. Finally, since $\varepsilon=n^{-\log_2{n}}$ and $m=O(n^2)$ we get that the last quantity {$m\varepsilon/2=O(n^2/n^{\log_2{n}})$} is negligible. So the input to the oracle ${\mathcal{O}}$ is statistically close to the uniform distribution on ${\mathcal{C}}_Q^{n\times m}$. 

Now, on input our matrix $A$ to the oracle, we get that the success probability is at least 
\[1/n^{c_0}-m\varepsilon/2>1/2n^{c_0} \Leftrightarrow \varepsilon < \frac{1}{2mn^{c_0}},\] where the inequality holds for large $n.$

$({\bf iv}).$ To get $n$ independent vectors of ${\mathcal{L}}$ after polynomially many calls to Algorithm \ref{alg:alg1} we need to prove that the output vector follows the shifted discrete Gaussian and then, apply Lemma \ref{lemma:indepedent_vectors} of the Appendix.
 We reverse the order of steps in Algorithm \ref{alg:alg1} as  (2)-(3)-(4)-(1)-(5). Namely, we first invoke the oracle in step 5  yielding a short integer vector ${\bf r} = (r_1,\ldots,r_m)$ with $\|{\bf r}\|_{\infty} \le \beta$ satisfying $A{\bf r} = {\bf 0}$. Then, the vector ${\bf x}_i$ is drawn as in step 1 but conditional to ${\bf y}_i\equiv {\bf x}_i\pmod{P({B})}.$ So
\[
{\bf v} = \sum_{i=1}^{m} r_i({\bf y}_i - {\bf x}_i)
\]
has the same distribution as in Algorithm \ref{alg:alg1}.
I.e. we apply the steps of Algorithm \ref{alg:alg1} with the following order:\\
$({\bf 1}).$ Choose ${\bf y}_i$ such that ${\bf y}_i\equiv D_{\mathcal{L},\tilde{\eta}}\pmod{P({B})}$ $(i=1,2,\dots,m).$\\
$({\bf 2}).$ Set $A$ be the matrix with columns ${\bf a}_i = \lfloor QB^{-1}{\bf y}_i \rfloor$ $(i=1,2,\dots,m).$\\
$({\bf 3}).$ With (non-negligible) probability $1/n^{c_0}$ we solve $AX={\bf 0}$ and say ${\bf r}=(r_1,\dots,r_m)\in {\mathbb{Z}}^m$ be a solution with 
$\|{\bf r}\|_{\infty}\le \beta.$\\
$({\bf 4}).$ Choose $m-$vectors ${\bf x}_i\xleftarrow{\$} {D_{{\mathcal{L}}+{\bf y}_i,\tilde{\eta}}}$  $(i=1,2,\dots,m).$\\
$({\bf 5}).$ The algorithm returns the lattice vector 
${\bf v}=\sum_{i=1}^{m}r_i({\bf y}_i-{\bf x}_i).$ \\\\
That is, we changed the steps of the original algorithm. In both algorithms the outputs follow the same distribution. Since now we now know the distribution of the outputs of Algorithm \ref{alg:alg1} we can apply Lemma \ref{lemma:indepedent_vectors} in order to get $n$ independent vectors of  ${\mathcal{L}}.$ 
\end{proof}
We proved Theorem \ref{theorem}. We now describe an explicit algorithm to choose the parameter $\tilde{\eta}$ within the interval \([2\eta_{\varepsilon}(\mathcal{L}), 4\eta_{\varepsilon}(\mathcal{L})]\). 

\begin{theorem}
    The  number of calls for Algorithm \ref{alg:alg1} in the reduction is at most 
    $\tilde{O}(n^{3+c_0}).$
\end{theorem}
\begin{proof}
   Algorithm~\ref{alg:alg1} requires a Gaussian parameter $\tilde{\eta}$ that is within
a constant factor of the smoothing parameter $\eta_{\varepsilon}(\mathcal{L})$
(for $\varepsilon=n^{-\log_2 n}$).

Compute an LLL-reduced basis $B_{\mathrm{LLL}}$ of $\mathcal{L}$ and let $R = \max_{1\le i\le n}\bigl\|{\bf b}^{\mathrm{LLL}}_i\bigr\|.$
Since $B_{\mathrm{LLL}}$ is a basis of $\mathcal{L}$, it contains $n$ linearly independent lattice vectors of norm at most $R$, hence $\lambda_n(\mathcal{L}) \le R$.  Together with Corollary~\ref{cor:smoothing-bound}, 
$\eta_{\varepsilon}(\mathcal{L}) \le 2\lambda_n(\mathcal{L})\log_2 n$. 
We set
\[
\widehat{\eta} = 2R\log_2 n
\]
satisfying $\eta_{\varepsilon}(\mathcal{L}) \le \widehat{\eta}.$

On the other hand, it is standard that (see Lemma \ref{Lemma:lll-vectors}) $R$ is upper-- bounded by $2^{n}\lambda_n(\mathcal{L})$. From Lemma \ref{Lemma:smoothing-bound2}, we get $\lambda_n<n\eta_{\varepsilon}({\mathcal{L}})$
so $\widehat{\eta}< 2^{n}{\rm{poly}}(n)\eta_{\varepsilon}({\mathcal{L}}).$

Thus we get,
\[\eta_\varepsilon({\mathcal{L}}) \le \widehat{\eta}\le \alpha\,
\eta_{\varepsilon}(\mathcal{L})\] for some $\alpha \le 2^{n}{\rm{poly}}(n)$.
Define the candidate values
\[
\tilde{\eta}_k = 2^{-k}\widehat{\eta}
\qquad (k=0,1,\ldots,K=\lceil \log_2{\alpha}\rceil +2),
\]
where $K= O(n+\log n)=\tilde{O}(n)$.
Then, there exists $k^\star\in\{0,\ldots,K\}$ such that,
\[
2\,\eta_{\varepsilon}(\mathcal{L}) \le \tilde{\eta}_{k^\star} \le
4\,\eta_{\varepsilon}(\mathcal{L}).
\]
  Since we have to run Algorithm \ref{alg:alg1} at most $n^{c_0+2}$ times for each of the $\tilde{O}(n)$ selections of $\tilde{\eta},$ the total calls are at most $\tilde{O}(n^{3+c_0}).$
\end{proof}

\section{Conclusion}\label{sec:conclusion}

In the present paper, we establish a worst-case to average-case connection for
$\mathrm{SIS}_{\mathbb Z}$. More precisely, we show that if random instances of
$\mathrm{SIS}_{\mathbb Z}$ can be solved in polynomial time with
non-negligible probability, then $\mathrm{SIVP}$ can be approximated in polynomial
time within a factor $\widetilde{O}(n^{3/2})$ in the worst case.

In any case, introducing new average-case problems with provable reductions 
to hard lattice problems can further advance our understanding of emerging post-quantum cryptography.  

As a next step, we plan to study a three-move identification scheme in the
Schnorr \& Lyubashevsky paradigm (see \cite{Vadim}; see also \cite{rizos}).

\appendix\label{A}
\section{}\label{sisz_sisq}
\textbf{Uniformity.} 
We check whether a matrix drawn from
\({\mathcal{C}}_Q^{n\times m}\) (entries i.i.d.\ uniform on \({\mathcal{C}}_Q=\{0,1,...,Q-1\}\)) remains
uniform modulo \(q\). 

Consider first the one-dimensional case \(n=m=1\).
Let \(f(x)=x \bmod q\) and let \(X\) be uniform on \({\mathcal{C}}_Q\); set \(Y=X\bmod q\).
Then for any \(a\in\mathbb Z_q\),
\[
\Pr[Y=a]
=\frac{\bigl|\{x\in {\mathcal{C}}_Q:\ x\equiv a \ (\mathrm{mod}\ q)\}\bigr|}{Q}
=\frac{\bigl|f^{-1}(\{a\})\cap {\mathcal{C}}_Q\bigr|}{Q}.
\]
First we prove the following Lemma.
\begin{lemma}\label{lemma:distribution-modq}
$Y$ is uniform over ${\mathbb{Z}}_q$ if and only if $q|Q.$
\end{lemma}
\begin{proof}
    $Y$ is uniform in ${\mathbb{Z}}_q$ $\Leftrightarrow$ 
    $$Pr[Y=a]=\frac{1}{q} \Leftrightarrow \frac{\bigl|\{x\in {\mathcal{C}}_Q:\ x\equiv a \ (\mathrm{mod}\ q)\}\bigr|}{Q}=\frac{1}{q}\Leftrightarrow$$
    $$\bigl|\{x\in {\mathcal{C}}_Q:\ x\equiv a \ (\mathrm{mod}\ q)\}\bigr|=\frac{Q}{q}\in {\mathbb{Z}}\Leftrightarrow q|Q.$$
\end{proof}
 We have the following Lemma.
\begin{lemma}
If $q\nmid Q,$ then
\[
0<\Delta\!\bigl(Y,\mathcal U(\mathbb Z_q)\bigr)
\le \frac{q}{4Q}.
\]
\end{lemma}
\begin{proof}
Let \(Q=qt+r\) with \(0<r<q\) (since $q\nmid Q$ we get $r>0$).
A standard  counting argument gives
\[
\Pr[Y=a]=
\begin{cases}
\dfrac{t+1}{Q}, & \text{if } a\in\{0,\ldots,r-1\},\\[6pt]
\dfrac{t}{Q}, & \text{otherwise.}
\end{cases}
\]
Set $\Delta=\Delta\bigl(Y,\mathcal U(\mathbb Z_q)\bigr).$
Since $q\nmid Q$, $Y$ is not uniform i.e. $\Delta> 0.$
    
    The statistical distance is
\[
\Delta=\frac12 \sum_{a\in\mathbb Z_q}\Bigl|\Pr[Y=a]-\tfrac1q\Bigr|
= \frac{1}{2}\bigg(r\Big|\frac{t+1}{Q}-\frac{1}{q}\Big|+(q-r)\Big|\frac{t}{Q}-\frac{1}{q}\Big|\bigg)=
\]
\[
\frac{1}{2}\bigg(r\Big|\frac{tq+q-Q}{Qq}\Big|+(q-r)\Big|\frac{tq-Q}{Qq}\Big|\bigg)=
\frac{r(q-r)}{Q\,q}.
\]
Finally, since $r(q-r)\le q^2/4$, we obtain
\[
\Delta\le \frac{q}{4Q}.
\]
\end{proof}

\section{}\label{upper_bound_for_smooth+parameter}
We provide the proof of Lemma \ref{Lemma:smoothing-bound2}.
In order to show that $\eta_{\varepsilon}({\mathcal{L}})>\sigma,$
it is enough to show that for $\varepsilon<1/100,$ there is $\sigma$ such that
\[
\rho_{1/\sigma}( \mathcal{L}^* \backslash \{\mathbf{0}\} )> \varepsilon.
\]
From Banaszczyk’s transference Theorem \cite{banaszczyk}, we know that
\[
1 \leq \lambda_n(\mathcal{L})\lambda_1(\mathcal{L}^*) \leq n.
\]
This can also be written as 
\begin{equation}\label{transference-theorem}
\frac{\lambda_n(\mathcal{L})}{n} \leq \frac{1}{\lambda_1(\mathcal{L}^*)}.    
\end{equation}Now, let $\sigma = \frac{1}{\lambda_1(\mathcal{L^*})}$ and $\mathbf{z}$ be a lattice point in $ \mathcal{L^*}$ with norm $\lambda_1({\mathcal{L^*}}).$ We know that 
\[
\rho_{1/\sigma}( \mathcal{L}^* \backslash \{\mathbf{0}\} )= \sum_{{\bf y} \in \mathcal{\mathcal{L}^* \backslash \{\mathbf{0}\}}} e^{-\pi \|{{\bf y}\|^2 \sigma^2}}
\]
and since $\mathbf{z} \in \mathcal{L}^* \backslash\{\mathbf{0}\}$, we get $\rho_{1/\sigma}( \mathcal{L}^* \backslash \{\mathbf{0}\} ) \ge \rho_{1/\sigma}(\mathbf{z}).$ We compute
\[
\rho_{1/\sigma}(\mathbf{z}) = e^{-\pi \|\mathbf{z}\|^2\sigma^2  } =  e^{-\pi \lambda_1^2(\mathcal{L}^*) / \lambda_1^2(\mathcal{L}^* )} = e^{-\pi}
\]
where $e^{-\pi} \approx 0.04322 > 1/100.$ So $\rho_{1/\sigma}( \mathcal{L}^* \backslash \{\mathbf{0}\} ) > 1/100,$ thus, $\eta_{\varepsilon}({\mathcal{L}})>\sigma$ (for $\varepsilon<1/100).$ Combining with inequality (\ref{transference-theorem}) and the fact $\frac{1}{\lambda_1(\mathcal{L}^*)}=\sigma$, we get
\[
\frac{\lambda_n(\mathcal{L})}{n} \leq \frac{1}{\lambda_1(\mathcal{L}^*)} < \eta_{\varepsilon}(\mathcal{L}).
\]
\end{document}